\documentclass{article}
\usepackage{spconf,amsmath,graphicx,amsthm}

\usepackage{amsmath,amsfonts,bm}

\def\eqref#1{equation~\ref{#1}}

\def\1{\bm{1}}

\def\rmC{{\mathbf{C}}}
\def\rmD{{\mathbf{D}}}

\def\rmI{{\mathbf{I}}}

\def\rmK{{\mathbf{K}}}

\def\rmP{{\mathbf{P}}}

\def\rmX{{\mathbf{X}}}
\def\rmY{{\mathbf{Y}}}

\def\va{{\bm{a}}}
\def\vb{{\bm{b}}}

\def\vp{{\bm{p}}}

\def\vu{{\bm{u}}}
\def\vv{{\bm{v}}}

\def\vx{{\bm{x}}}
\def\vy{{\bm{y}}}

\DeclareMathAlphabet{\mathsfit}{\encodingdefault}{\sfdefault}{m}{sl}
\SetMathAlphabet{\mathsfit}{bold}{\encodingdefault}{\sfdefault}{bx}{n}

\def\sR{{\mathbb{R}}}

\usepackage{algpseudocode,algorithm,algorithmicx} 
\usepackage[numbers]{natbib}

\newtheorem{proposition}{Proposition}

\title{Wasserstein total variation filtering}
\name{Erdem Varol, Amin Nejatbakhsh\thanks{Funding sources: NSF NeuroNex Award DBI-1707398, ~The Gatsby Charitable Foundation.}}
\address{Columbia University, Department of Statistics\\
    Zuckerman Institute, Center for Theoretical Neuroscience}
\begin{document}
%
\maketitle
\begin{abstract}
In this paper, we expand upon the theory of trend filtering by introducing the use of the Wasserstein metric as a means to control the amount of spatiotemporal variation in filtered time series data. While trend filtering utilizes regularization to produce signal estimates that are piecewise linear, in the case of $\ell_1$ regularization, or temporally smooth, in the case of $\ell_2$ regularization, it ignores the topology of the spatial distribution of signal. By incorporating the information about the underlying metric space of the pixel layout, the Wasserstein metric is an attractive choice as a regularizer to undercover spatiotemporal trends in time series data. We introduce a globally optimal algorithm for efficiently estimating the filtered signal under a Wasserstein finite differences operator. The efficacy of the proposed algorithm in preserving spatiotemporal trends in time series video is demonstrated in both simulated and fluorescent microscopy videos of the nematode \textit{caenorhabditis elegans} and compared against standard trend filtering algorithms.
\end{abstract}
\begin{keywords}
optimal transport, total variation, motion filtering, wasserstein distance, trend filtering
\end{keywords}
\section{Introduction}\label{sec:intro}
Trend filtering aims to estimate the underlying signal in noisy time series data by decomposing the time series into a smooth component plus a randomly varying noise component. Different types of trend filters that have been proposed\cite{whittaker1922new,hodrick1997postwar,tibshirani2014adaptive,kim2009ell_1,rudin1992nonlinear} in the past have found uses in many applications in the fields of economics~\cite{hodrick1997postwar}, financial time series~\cite{tsay2014financial}, geophysics~\cite{bloomfield1992climate}, and biology~\cite{link1998estimating}, to name a few.

In the context of spatiotemporal data such as video time series, standard trend filtering cannot utilize the spatial contiguity of objects in view and essentially treats the video as a collection of independent pixelwise time series. Therefore, whilst able to handle temporal noise in terms of random jitter, traditional trend filtering methods do not model noise in the context of spatial movement jitter.

The recent popularity of the Wasserstein metric, also known as the earth mover's distance (EMD), as an alternative to the euclidean metric in the context of regularizer of loss functions in image processing allows an intuitive and efficient means of handling spatial motion in videos~\cite{kantorovich2006translocation,villani2008optimal,sinkhorn1967diagonal,ryu2018vector,peyre2019computational,cuturi2013sinkhorn,chen2018optimal,alvarez2017structured,arjovsky2017wasserstein}.

In this work, we propose to modify the traditional trend filtering frameworks to utilize a Wasserstein metric regularizer as a way to filter out not only temporal jitter but also noise in the form of excessive motion. The ability of the Wasserstein metric to naturalistically model physical properties of objects in terms of mass preservation and motion modelling have been demonstrated widely in~\cite{peyre2019computational}. We make use of the recently introduced entropy regularized Wasserstein metric\cite{cuturi2013sinkhorn,altschuler2017near} to form a strongly convex objective function that can be optimized to its global optimal efficiently with the proposed algorithm. We demonstrate the efficiency of the proposed filtering algorithm to denoise the movement of simulated spatiotemporal data. Furthermore, we showcase the use of the proposed method in a computational neuroscience context by denoising a fluorescent microscopy video of the neurons of the nematode \textit{caenorhabditis elegans} (\textit{c. elegans}).
\vspace{-4ex}
\subsection{Paper organization}
In sections~\ref{sec:l2tf} and~\ref{sec:l1tf}, we revisit various formulations of trend filtering to motivate the formulation for the proposed method. In section~\ref{sec:wtv} we introduce the proposed formulation and provide a globally convergent algorithm for its optimization. Section~\ref{sec:results} demonstates the filtering efficacy of the proposed method in comparison to trend filtering methods. In section~\ref{sec:discussion}, we discuss generalizations of the proposed method.
\section{Methods}\label{sec:method}
First we introduce notation. Let $\rmX \in \sR^{d\times T}$ denote the $d$-dimensional time series signal in $T$ temporal segments. The signal $\rmX$ can be denoted both in its temporal columns and its spatial rows as $\rmX = [\rmX_1 | \ldots | \rmX_T] = [\vx_1^T | \ldots | \vx_d^T]^T$. Likewise, let $\rmY\in \sR^{d\times T}$ denote the filtered estimator that is to be optimized. Likewise, we can also denote $\rmY$ by its columns or rows as $\rmY = [\rmY_1 | \ldots | \rmY_T] = [\vy_1^T | \ldots | \vy_d^T]^T$.

\subsection{Hodrick–Prescott filtering ($\ell_2$ trend filtering)}\label{sec:l2tf}
The trend filtering introduced by Hodrick and Prescott in~\cite{hodrick1997postwar} involves optimizing for an estimator $\rmY$ such that the estimator at time $t$ is close in euclidean metric to the average of the estimators at times $t+1$ and $t-1$ whilst respecting data fidelity:

\begin{align}\label{eq:l2tf}
    \underset{\rmY}{\min}\frac{1}{2}\sum_{t=1}^T \|\rmX_t - \rmY_t\|_2^2 + \lambda \sum_{t=1}^{T-1} \|\rmY_{t-1} -2\rmY_t + \rmY_{t+1}\|_2^2
\end{align}
If we denote $\rmD^{(1)} \in \sR^{T \times T}$ as the first order finite difference operator of the form
\begin{align*}
    \rmD^{(1)} = \left[\begin{matrix} 1 & -1 & & &  \\
    & 1 & - & &  &\\
    & &  & \ddots & &\\
      & & & 1 & -1& \\
    \end{matrix}\right]
\end{align*}
we can obtain $k$th order finite differences by composing the operator $\rmD^{(1)}$ with itself $k$ times: $\rmD^{(k)} = \rmD^{(1)}\rmD^{(k-1)}$. Given this, the regularizer in the Hodrick-Prescott filter can expressed as the squared euclidean norm of the product of the second order finite difference operator $\rmD^{(2)}$ with the rows of the estimator $\rmY$. This can be used to decompose $\ell_2$-trend filtering into separable problems along the rows of the signal matrix as
\begin{align}
    \underset{\vy_i}{\min}\frac{1}{2}\|\vx_i -\vy_i\|_2^2 + \lambda \|\rmD^{(2)}\vy_i\|_2^2 
\end{align}
which yields a closed form solution:
\begin{align}
    \vy^{(\ell_2)}_i = (\rmI + 2\lambda {\rmD^{(2)}}^T\rmD^{(2)})^{-1}\vx_i 
\end{align}
\subsection{Kim et al. $\ell_1$ trend filtering}\label{sec:l1tf}
In contrast with the Hodrick–Prescott filter, $\ell_1$ trend filter was introduced in \cite{kim2009ell_1} and expanded in \cite{tibshirani2014adaptive}. Instead of penalizing the squared euclidean distance of finite differences of the estimators $\rmY$, this formulation penalizes absolute differences:
\begin{align}\label{eq:l1tf}
    \underset{\rmY}{\min}\frac{1}{2}\sum_{t=1}^T \|\rmX_t - \rmY_t\|_2^2 + \lambda \sum_{t=1}^{T-1} \|\rmY_{t-1} -2\rmY_t + \rmY_{t+1}\|_1
\end{align}
which can similarly be decomposed into separable problems along the rows of the signal matrix as
\begin{align}
    \underset{\vy_i}{\min}\frac{1}{2}\|\vx_i -\vy_i\|_2^2 + \lambda \|\rmD^{(2)}\vy_i\|_1 
\end{align}
The optimum for this does not admit a closed form solution, however proximal gradient steps can be used to efficiently obtain a solution~\cite{tibshirani2014adaptive}. The intuition behind the $\ell_1$ trend filter is that it yields piecewise linear components which can then be used to estimate state changes in time series data. 

Note that if we utilize the first order finite differences operator in the above formulations, the resulting objective is that of minimizing the total variation denoiser~\cite{rudin1992nonlinear}. However, in all three cases, the prevailing trend is that the series along each spatial dimension is independent that of the other dimensions. In the context of video data, this assumption is clearly violated which motivates the formulation we introduce herein.

\begin{algorithm}[!htb]
\caption{Wasserstein total variation (WTV)}\label{alg:wtv}
\begin{algorithmic}
\State{\textbf{Input: }} Signal: $\rmX \in \sR^{d \times T}$, ground cost: $\rmC \in \sR^{d \times d}$, Wasserstein regularization: $\lambda \geq 0$, entropic regularization: $\gamma > 0$, convergence tolerance $\epsilon>0$, Sinkhorn iterations: $S >0$, step-size $\alpha >0$
\State{\textbf{Require:} Data normalization: $\rmX_t^T\boldsymbol{1} = 1$, $\rmX \geq 0$}
\State{\textbf{Initialization:}} $k=0$,~~$\rmY^{(k)} \leftarrow \rmX$,~~$\rmK = e^{-\lambda \rmC/\gamma}$,\\~~$\lbrace\va_t^{(0)}=\boldsymbol{1}_d,\vb_t^{(0)}=\boldsymbol{1}_d\rbrace_{t=1}^{T-1}$
\While{Not converged}
\State{$k\leftarrow k+1$}
\For{$t=1,\ldots,T-1$}
\State{Set $\vu_t^{(0)} = e^{-\va_t^{(k-1)}}$,~ $\vv_t^{(0)} = e^{-\vb_t^{(k-1)}}$}
\State{$s \leftarrow 0$}
\For{$s=1,\ldots,S$}
\State{$s \leftarrow s+1$}
\State{$\vu_t^{(s)} \leftarrow \rmY_t/\rmK \vv_t^{(s-1)} $ (element wise)}
\State{$\vv_t^{(s)} \leftarrow \rmY_t/\rmK \vu_t^{(s)} $ (element wise)}
\EndFor
\State{$\va_t^{(k)} \leftarrow -\gamma \log \vu_t^{(S)}$}
\State{$\vb_t^{(k)} \leftarrow -\gamma \log \vv_t^{(S)}$}
\State{$\rmP_t^{(k)} \leftarrow \text{diag}(e^{-\va_t^{(k)}/\gamma}) e^{-\lambda\rmC/\gamma}\text{diag}(e^{-\vb_t^{(k)}/\gamma})$}
\State{Projected gradient descent such that $\rmY^{(k)} \geq 0$:}
\State{$\rmY_t^{(k)} \stackrel{\geq 0}{\leftarrow}\rmY_t^{(k-1)}+ \alpha[\rmX_t - \rmY_t^{(k-1)} + \va_t^{(k)} + \vb_{t-1}^{(k)}]$}
\EndFor
\State{Check convergence:  $\|\rmY^{(k)}-\rmY^{(k-1)}\|_2 \stackrel{?}{\le} \epsilon$}
\EndWhile\\
\Return{Filtered signal: $\rmY\in\sR^{d \times T}$}
\end{algorithmic}
\end{algorithm}
\subsection{Wasserstein total variation}\label{sec:wtv}

\begin{figure*}
    \centering
    \includegraphics[width=1\linewidth]{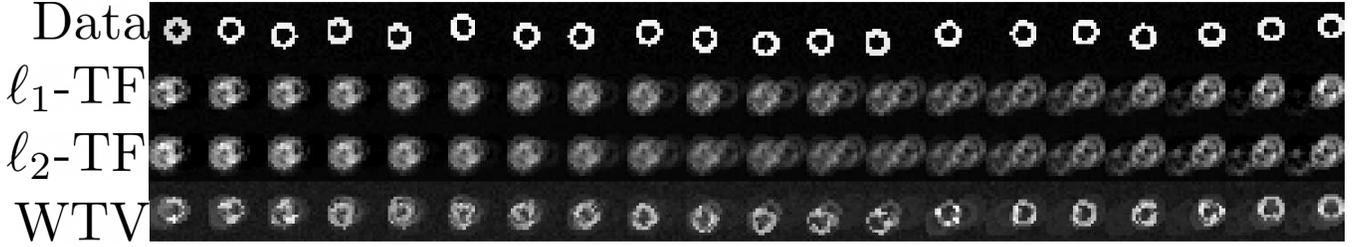}
    \caption{Results of time series filtering on simulated ring data over 20-time frames. Rows denote the different compared methods while columns denote temporal slices. The top row denotes the spatial and temporal variation of the simulated data over the time course. Second and third rows denote the results of $\ell_1$\cite{kim2009ell_1} and $\ell_2$\cite{hodrick1997postwar} trend filtering, respectively. The last row denotes the proposed Wasserstein total variation (WTV) filtering of the time series. Note the temporal averaging that $\ell_1$-TF and $\ell_2$-TF introduces which yields in low contrast filtering. In contrast, WTV allows  the warping of the underlying spatial layout to yield higher contrast time slices.}
    \label{fig:simulated}
\end{figure*}

Similar to trend filtering techniques, we seek to obtain a projection of the video frames to an equal dimension latent space $\rmY \in \sR^{d\times T}$ such that the Wasserstein distance between subsequent latent frames is minimized with a trade-off with respect to data fidelity. In its most rudimentary form, regularizing by first-order Wasserstein finite differences results in the following objective:

\begin{align}\label{wasserstein}
    \underset{\rmY}{\min} \frac{1}{2}\sum_{t=1}^T \|\rmX_t - \rmY_t\|_2^2 + \lambda \sum_{t=1}^T W_1(\rmY_t,\rmY_{t+1})
\end{align}
where the Wasserstein metric $W_1$ can be expressed as the solution to a constrained linear program:
\begin{align}
    &W_1(\rmY_t,\rmY_{t+1}) = \underset{\rmP}{\min}~\sum_{i,j=1}^{d} c_{i,j}P_{i,j}  \nonumber \\
    &\text{subject to}~\sum_{j=1}^d P_{i,j} = \rmY_t,~~\sum_{i=1}^d P_{i,j} = \rmY_{t+1} 
\end{align}

Here $c_{i,j}$ denotes the ground cost of transporting a unit of mass from coordinate $i$ to coordinate $j$. In practice, $c_{i,j}$ can be set to be the euclidean distance between pixel locations $i$ and $j$: $c_{i,j} = \|\vp_i - \vp_j\|_2$ where $\vp_i$ denotes the coordinates of the $i$th pixel.

These two objectives can be expressed jointly as
\begin{align}
    &\underset{\rmY,\rmP}{\min} \frac{1}{2}\sum_{t=1}^T \|\rmX_t - \rmY_t\|_2^2 + \lambda \sum_{t=1}^T \sum_{i,j=1}^{d} c_{i,j,t}P_{i,j,t} \nonumber \\
    &\text{subject to}\sum_{j=1}^d P_{i,j,t} = \rmY_t \nonumber \\
    &\sum_{i=1}^d P_{i,j,t} = \rmY_{t+1} 
\end{align}

This problem is a quadratic program with linear constraints and can be solved using off-the-shelf solvers. Due to the dimensionality of the problem, the complexity of obtaining the Wasserstein metric is in the order $O(d^3)$. Furthermore, Cuturi et al.\cite{cuturi2013sinkhorn} have shown that the minimizer of \eqref{wasserstein} is not unique due to the weak convexity of linear programs. Also, in recent works~\cite{altschuler2017near} it is shown that fixed point Sinkhorn iterations can be used to compute the metric in $O(d^2)$ time with logarithmic scaling with the added benefit of providing strong convexity which guarantees a unique optimal solution. The objective with entropic regularization becomes:

\begin{align}\label{eq:wtv_objective}
   & \underset{\rmY,\rmP}{\min} \frac{1}{2}\sum_{t=1}^T \|\rmX_t - \rmY_t\|_2^2 + \lambda \sum_{t=1}^T \sum_{i,j=1}^{d} c_{i,j,t}P_{i,j,t}  +  \gamma H(\rmP_t) \nonumber\\
   & \text{subject to}\sum_{j=1}^d P_{i,j,t} = \rmY_t \nonumber \\
   &\sum_{i=1}^d P_{i,j,t} = \rmY_{t+1} 
\end{align}

where $H(\rmP_t) = \sum_{i,j}^d P_{i,j,t} \log P_{i,j,t} - 1$ denotes the entropy of the transportation matrix $\rmP_t$.

Taking gradients of the Lagrangian of \eqref{eq:wtv_objective} with dual variables $\va_t$ and $\vb_t$ yields the following Karesh-Kuhn-Tucker (KKT) conditions:
\begin{align}\label{eq:kkt}
    \frac{\partial \mathcal{L}}{\partial \rmY_t} &= - \rmY_t -\rmX_t - \va_t - \vb_{t-1} = 0\nonumber \\
    \frac{\partial \mathcal{L}}{\partial \rmP_{i,j,t}} &= \lambda c_{i,j,t} +  \gamma\log P_{i,j,t} + \va_{i,t} + \vb_{j,t} = 0 \nonumber \\
    \frac{\partial \mathcal{L}}{\partial \va_t} &= \rmP_t \boldsymbol{1} - \rmY_t =0 \nonumber \\
    \frac{\partial \mathcal{L}}{\partial \vb_t} &= \rmP_t^T \boldsymbol{1} - \rmY_{t+1}=0
\end{align}

\begin{figure*}[!htb]
    \centering
    \includegraphics[width=1\linewidth]{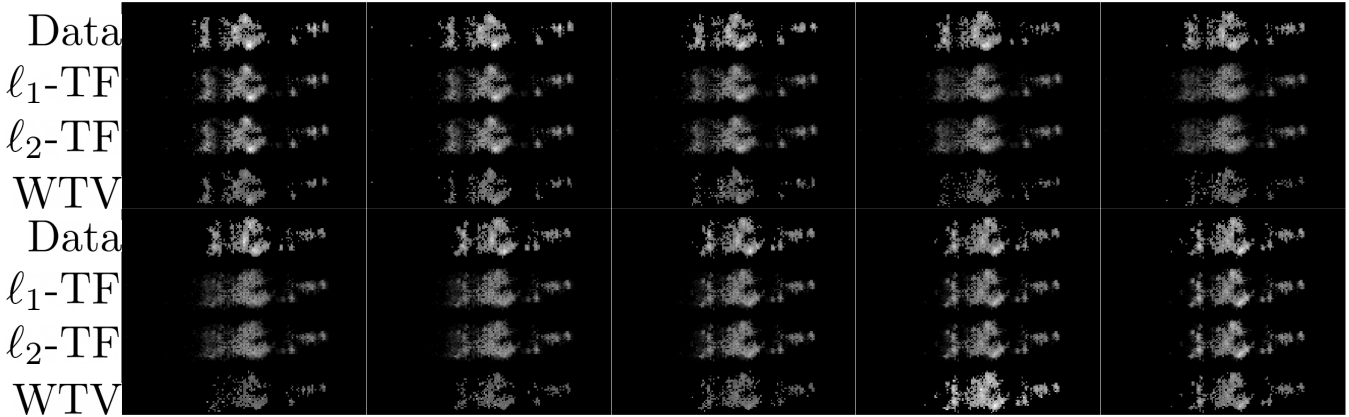}
    \caption{Results of time series filtering on \textit{c.elegans} fluorescence microscopy over 100-time frames (middle 10 frames shown sequentially in two rows). Rows denote the different compared methods while columns denote temporal slices. The top row denotes the spatial and temporal variation of the simulated data over the time course. Second and third rows denote the results of $\ell_1$\cite{kim2009ell_1} and $\ell_2$\cite{hodrick1997postwar} trend filtering, respectively. The last row denotes the proposed Wasserstein total variation (WTV) filtering of the time series. Note that WTV allows the warping of the worm to yield spatiotemporal trends.}
    \label{fig:worm}
\end{figure*}
The optimal $P_{i,j,t}$ can be expressed as
\begin{align}
    \rmP_t^* = \text{diag}(e^{-\va_t/\gamma}) e^{-\lambda\rmC/\gamma}\text{diag}(e^{-\vb_t/\gamma})
\end{align}
which by using the KKT conditions result in the following system:
\begin{align}
    \rmP_t^*\boldsymbol{1} &= \vu_t \odot \rmK \vv_t = \rmY_t \nonumber \\
    {\rmP_t^*}^T\boldsymbol{1} &= \vv_t \odot \rmK \vu_t = \rmY_{t+1} 
\end{align}
where $\odot$ denotes elementwise multiplication. The auxillary dual variables $\vu_t = e^{-\va_t/\gamma} $ and $\vv_t = e^{-\vb_t/\gamma}$ can then be iteratively estimated using the convergent fixed point Sinkhorn-Knopp\cite{sinkhorn1967concerning} iterations
\begin{align}\label{eq:sinkhorn}
    \vu_t \leftarrow \rmY_t/\rmK \vv_t,~~~
    \vv_t \leftarrow \rmY_{t+1}/\rmK \vu_t
\end{align}
Algorithm~\ref{alg:wtv} summarizes the optimization routine outlined above that is used to infer the estimator $\rmY$.

\begin{proposition}
Algorithm~\ref{alg:wtv} converges to a unique global optimum.
\end{proposition}
\begin{proof}
This follows from the fact that the objective in \eqref{eq:wtv_objective} is strongly convex due to inclusion of the entropic regularizer term $H(\rmP_t)$ and that the constraint set is affine on the optimized variables. Furthermore, \cite{sinkhorn1967concerning} showed that given sufficient number of iterations, taking the logarithm of the iterates in procedure in~\eqref{eq:sinkhorn} yield the dual variables that satisfy the KKT conditions in ~\eqref{eq:kkt}.
\end{proof}

\section{Experiments}\label{sec:results}
To demonstrate the efficacy of the proposed algorithm on denoising spatiotemporal data, we perform experiments on both simulated data as well as fluorescence microscopy videos of the nematode \textit{c. elegans}. The proposed WTV algorithm is compared with both $\ell_2$~\cite{hodrick1997postwar} and $\ell_1$~\cite{kim2009ell_1,tibshirani2014adaptive} trend filtering. The $\lambda$ parameters of the trend filtering methods were set such that the data fidelity term $\sum_t\|\rmX_t-\rmY_t\|_2^2$ yielded the same loss for both methods. Similarly, for WTV, the entropic regularization term $\gamma$ was set to be 1 and $\lambda$ was set such that the data fidelity loss term equaled that of the trend filters.

\subsection{Simulated data}
The simulated data was generated by rendering a ring shape in a 256 x 256 space and temporally shifting its center using a random walk with a variance of 25 pixels. 20-time frames were generated using this procedure. The data can be visualized in the top row of figure ~\ref{fig:simulated}. The results of the trend filters can be seen in the second and third rows. Lastly, WTV results are seen in the last row of figure~\ref{fig:simulated}. The results show that indeed since trend filtering is not informed by the geometry of the image layout, it yields estimators that are blurry due to temporal averaging. In contrast, WTV yields an estimator that is a warped version of the underlying signal such that motion between frames is smoothed.
\subsection{C. elegans data}
The fluorescence microscopy data of \textit{c.elegans}~\cite{yemini2019neuropal} consisted of 100 time frames of spatial dimensions 116 x 600 x 87 voxels. The images were downsampled by a factor of 4 and background pixels were removed by Otsu's method~\cite{otsu1979threshold}. The results for this experiment are seen in figure~\ref{fig:worm}. Similar to the simulated experiments, the results of WTV differ from those of trend filters by yielding images that are averaged spatially instead of temporally. Note that WTV allows the warping of the worm to yield spatiotemporal trends.

\section{Discussion}\label{sec:discussion}
\subsection{Generalized Wasserstein finite differences}
The signal filtering model we introduce is analogous to performing trend filtering using a first-order finite-difference operator in the Wasserstein metric space. Since first-order finite-difference regularization in the euclidean metric space is also known as total variation regularization\cite{rudin1992nonlinear}, we term the method herein as Wasserstein total variation filtering. Generalizing the method we propose by taking higher-order finite differences in the Wasserstein metric space corresponds to the computation of Wasserstein barycenters\cite{cuturi2014fast}. Hence, the second-order Wasserstein trend filter would consist of the following regularizer: $ W_1(\rmY_t,BW_1(\rmY_{t+1},\rmY_{t-1}))$
where $BW_1(\cdot,\cdot)$ denotes the two input barycenter operator in the Wasserstein metric space such that $W_1(BW(A,B),A) = W_1(BW(A,B),B)$. Deriving efficient algorithms in this context remains an open problem for future work.\\
\subsection{Conclusion}
In this paper, we introduced a non-parametric estimation model for the filtering of spatiotemporal trends in time series data. Utilizing regularization in the Wasserstein metric space in contrast with the euclidean metric space allows the proposed filter to harness spatial correspondences across time frames to yield an estimator that is not only steadfast to the temporal trends in the signal but also respects the underlying motion of the objects in view. The simulated experiments showcase the ability of the filter to warp the data in a way that preserves the underlying spatial movement of the signal. Furthermore, results on fluorescence microscopy videos of \textit{c. elegans} nematodes show that the proposed method can reduce motion artifacts in low framerate videos commonly encountered in computational biology contexts.

\bibliographystyle{IEEEbib}
\bibliography{refs}

\end{document}